\documentclass[11pt]{amsart}
\usepackage{amsmath,amssymb,amsthm,a4wide}
\usepackage{latexsym}
\usepackage{indentfirst}
\usepackage{placeins}
\usepackage{booktabs}
\usepackage{algorithm}
\usepackage{algorithmic}
\usepackage{multirow}
\usepackage{color,xcolor,colortbl}

\usepackage{subfigure}
\usepackage{graphicx,color}
\usepackage{diagbox}

\usepackage{hyperref}

\usepackage{cleveref}
\usepackage{braket}
\usepackage{nicematrix}

\usepackage{pifont}
\usepackage{multirow}

\theoremstyle{plain}
\newtheorem{theorem}{Theorem}
\newtheorem{lemma}[theorem]{Lemma}

\theoremstyle{definition}

\theoremstyle{definition}

\theoremstyle{remark}

\DeclareMathOperator{\polylog}{polylog}
\DeclareMathOperator{\poly}{poly}

\usepackage{qcircuit}

\newcommand{\eps}{\epsilon}

\newcommand{\abs}[1]{\left|#1\right|}

\newcommand{\norm}[1]{\left\lVert#1\right\rVert}

\newcommand{\bbm}{\begin{bmatrix}}
	\newcommand{\ebm}{\end{bmatrix}}
\newcommand{\RR}{\mathbb{R}}
\newcommand{\CC}{\mathbb{C}}
\newcommand{\DD}{\mathbb{D}}

\newcommand{\TT}{\mathbb{T}}
\newcommand{\ZZ}{\mathbb{Z}}

\newcommand{\T}{\mathsf{T}}

\renewcommand{\T}{T}

\renewcommand{\i}{\mathrm{i}}

\newcommand{\diag}{\mathrm{diag}}

\newcommand{\I}{\mathrm{i}}
\newcommand{\half}{\frac{1}{2}}
\newcommand{\mo}{\mathcal{O}}
\newcommand{\tmo}{\tilde{\mathcal{O}}}

\newcommand{\ba}{\mathbf{a}}
\newcommand{\bb}{\mathbf{b}}
\newcommand{\bc}{\mathbf{c}}

\newcommand{\be}{\mathbf{e}}

\newcommand{\bl}{\mathbf{l}}

\newcommand{\bp}{\mathbf{p}}
\newcommand{\bq}{\mathbf{q}}

\newcommand{\bu}{\mathbf{u}}
\newcommand{\bv}{\mathbf{v}}
\newcommand{\bw}{\mathbf{w}}

\newcommand{\by}{\mathbf{y}}

\newcommand{\rev}{\mathrm{rev}}
\newcommand{\bzero}{\mathbf{0}}
\newcommand{\xmark}{\textcolor{black}{\ding{55}}}
\newcommand{\cmark}{\textcolor{black}{\ding{51}}}

\newcommand{\hg}{\hat{G}}
\newcommand{\hr}{\hat{R}}
\newcommand{\hq}{\hat{Q}}
\newcommand{\hu}{\hat{\bu}}
\newcommand{\hv}{\hat{\bv}}
\newcommand{\hU}{\hat{U}}

\newcommand{\hL}{\hat{L}}
\newcommand{\hh}{\hat{H}}
\newcommand{\ema}{\epsilon_{\mathrm{ma}}}

\title{Fast Phase Factor Finding for Quantum Signal Processing}

\author[]{Hongkang Ni} \address[Hongkang Ni]{Stanford University,
  Stanford, CA 94305} \email{hongkang@stanford.edu}

\author[]{Lexing Ying} \address[Lexing Ying]{Stanford University,
  Stanford, CA 94305} \email{lexing@stanford.edu}

\thanks{This work is partially supported by NSF grant DMS-2208163. Support is also acknowledged from the U.S. Department of Energy, Office of Science, Accelerated Research in Quantum Computing Centers, Quantum Utility through Advanced Computational Quantum Algorithms, grant no. DE-SC0025572. We thank Lin Lin and Christoph Thiele for helpful discussions.}


\begin{document}
	
\begin{abstract}
	This paper presents two efficient and stable algorithms for recovering phase factors in quantum signal processing (QSP), a crucial component of many quantum algorithms. The first algorithm, the ``Half Cholesky" method, which is based on nonlinear Fourier analysis and fast solvers for structured matrices, demonstrates robust performance across all regimes. The second algorithm, ``Fast Fixed Point Iteration," provides even greater efficiency in the non-fully-coherent regime. Both theoretical analysis and numerical experiments demonstrate the significant advantages of these new methods over all existing approaches.
\end{abstract}

\maketitle

\section{Introduction}

\subsection{Background and Problem Setup}
Quantum Signal Processing (QSP) is widely used in various quantum algorithms, providing quantum computers with greater flexibility in addressing linear algebra problems, such as Hamiltonian simulation \cite{LowChuang2017,GilyenSuLowEtAl2019}, linear system and PDE solvers \cite{GilyenSuLowEtAl2019,LinTong2020,MartynRossiTanEtAl2021,li2023efficient}, eigenvalue problems \cite{LinTong2020a,DongLinTong2022}, Gibbs state preparation \cite{GilyenSuLowEtAl2019}, quantum system benchmarking \cite{DongLin2021,DongWhaleyLin2021}, and quantum device calibration \cite{niu2024quantum}, among others. For a more comprehensive review of QSP applications, we refer interested readers to \cite{GilyenSuLowEtAl2019, MartynRossiTanEtAl2021}.

To introduce the mathematical framework underlying QSP, we begin by defining two unitary matrices parameterized by \( x \in [-1, 1] \):
\begin{equation}
	W(x) = \begin{pmatrix}
		x & \I \sqrt{1-x^2} \\
		\I \sqrt{1-x^2} & x
	\end{pmatrix},
	\quad \text{and} \quad \sigma_z = \begin{pmatrix}
		1 & 0 \\
		0 & -1
	\end{pmatrix}.
\end{equation}

Given a polynomial \( f(x) \in \mathbb{R}[x] \) of degree \( n \) with parity \( n \bmod 2 \), defined on the interval \([-1, 1]\), and satisfying
\begin{equation}
	\|f\|_{\infty} = \max_{x \in [-1, 1]} |f(x)| < 1,
\end{equation}
the task of Quantum Signal Processing (QSP) is to find a sequence of phase factors \( \Phi = (\phi_0, \ldots, \phi_n) \in [-\pi, \pi]^{n+1} \) such that \( f(x) = \operatorname{Im}(p(x)) \), where \( p(x) \) is the upper left entry of the unitary matrix
\begin{equation}\label{eq: qsp matrix}
U(x, \Phi) = \begin{pmatrix}
	p(x) & * \\
	* & *
\end{pmatrix} = e^{i \phi_0 \sigma_z} W(x) e^{i \phi_1 \sigma_z} W(x) \cdots e^{i \phi_{n-1} \sigma_z} W(x) e^{i \phi_n \sigma_z}.
\end{equation}
We denote \( g(x, \Phi) := \operatorname{Im}(p(x)) \).

The problem becomes more challenging as \( \|f\|_{\infty} \) approaches 1. To capture this difficulty, for each \( \eta \in (0,1) \), we define the set
\begin{equation}
	\mathbf{S}_{\eta} = \{f : \|f\|_{\infty} \leq 1 - \eta\},
\end{equation}
where the regime of small \( \eta \) is known as the \emph{fully-coherent} regime \cite{dong2023robust}. In the non-fully-coherent regime, which is also common in certain quantum computing tasks, scaling down \( f \) by a constant to increase $\eta$ is permissible \cite{DongLinNiEtAl2022}. However, even in this regime, it is preferable to avoid $\norm{f}_{\infty}$ to be too small. 

Due to the parity constraint, the number of degrees of freedom in the target polynomial \( f \in \mathbb{R}[x] \) is only \( \lceil \frac{n+1}{2} \rceil \). Therefore, the phase factors \( \Phi \) cannot be uniquely defined. To address this issue, Ref.~\cite{DongMengWhaleyEtAl2021} suggests imposing a symmetry constraint on the phase factors:
\begin{equation}
	\phi_j = \phi_{n-j}, \quad \forall j = 0, 1, \ldots, n.
	\label{eqn:symmetry_phase}
\end{equation}
Let \( d = \lceil \frac{n+1}{2} \rceil - 1 \), and define the \emph{reduced phase factors} as follows:
\begin{equation}
	\Psi = (\phi_0, \phi_1, \ldots, \phi_d) := \begin{cases}
		(\phi_d, \phi_{d+1}, \ldots, \phi_n), & \text{if } n \text{ is even}, \\
		(\phi_{d+1}, \phi_{d+2}, \ldots, \phi_n), & \text{if } n \text{ is odd}.
	\end{cases}
\end{equation}
With some abuse of notation, we identify \( U(x, \Psi) \) with \( U(x, \Phi) \), and \( g(x, \Psi) \) with \( g(x, \Phi) \).

The existence of symmetric phase factors was established in \cite[Theorem 1]{WangDongLin2022}, but the solution remains non-unique, though finite in number. However, near the trivial phase factors \( \Psi = (0, \ldots, 0) \), there exists a unique and consistent choice of symmetric phase factors called the maximal solution~\cite{WangDongLin2022}. In this paper, unless otherwise specified, the term ``phase factor'' will always refer to this maximal solution.

Without loss of generality, we restrict our discussion to the case where \( n = 2d \) is even throughout the paper. A similar treatment can be extended to the odd case. We note that our notations and conventions are generally consistent with \cite{alexis2024infinite}, but differ from \cite{DongLinNiEtAl2022}.


\subsection{Main results} 

We developed a provably stable algorithm called ``Half Cholesky'' that is the most efficient among all existing methods for finding the QSP phase factors in all regimes. This algorithm is based on the nonlinear Fourier transform (NLFT) and fast solvers for structured matrices. The main result is summarized in the following theorem. The details of the algorithm are discussed in \Cref{sec: fast NLFT}, with complexity and stability analyzed in \Cref{sec: HC complexity} and \Cref{sec: HC stability}, respectively.

\begin{theorem}\label{thm: RHW_main_thm}
    Assume that $f$ is a $(2d)$-degree even polynomial in $\mathbf{S}_{\eta}$. There exists a deterministic algorithm (\Cref{alg: phase factor finding}) to stably compute the phase factor sequence $\Psi$ to precision $\epsilon$ with a computational cost of $\tmo\left(d^2 + \frac{d\log(d/(\eta\epsilon))}{\eta}\right)$.
\end{theorem}

Additionally, we have developed an even faster algorithm called ``Fast Fixed Point Iteration'' for finding QSP phase factors, with complexity nearly linear in the problem size \( d \). This algorithm leverages a fixed-point iteration method combined with a divide-and-conquer technique to accelerate the computation of the QSP matrix \eqref{eq: qsp matrix}. While this approach does not apply to the fully-coherent regime, it is effective for a wide range of target functions \( f \), especially those relevant to common quantum computing tasks. In the non-fully-coherent regime, this is the most efficient algorithm among all existing methods for phase factor finding. The result is summarized in the following theorem, with details discussed in \Cref{sec: fast FPI}.

\begin{theorem}\label{thm: FPI_main_thm}
    Let $f$ be a $(2d)$-degree even polynomial with Chebyshev expansion $f(x) = \sum_{j=0}^{d} \hat{f}_j T_{2j}(x) = \sum_{j=0}^{d} \hat{f}_j \cos(2j\arccos(x))$. If $\|\hat{f}\|_1 := \sum_{j=1}^{d} |\hat{f}_j| < 0.861$, then there exists a deterministic algorithm (see \eqref{eq: FPI iter}) to stably compute the phase factor sequence $\Psi$ to precision $\epsilon$ with a computational cost of $\mo\left(d\log^2 d \log\eps^{-1}\right)$.
\end{theorem}

\subsection{Related works}

Quantum Signal Processing (QSP) was introduced for quantum computing in \cite{LowChuang2017}, where the existence of phase factors was proven non-constructively. Since then, various methods have been developed to efficiently and stably determine the phase factors given a target polynomial $f$. We summarize the features and computational complexities of these methods in \Cref{table: compare}. 

The earliest methods \cite{GilyenSuLowEtAl2019, Haah2019, ChaoDingGilyenEtAl2020} are based on root-finding procedures to determine the missing components of the QSP matrix \eqref{eq: qsp matrix}, referred to as complementary polynomials. However, these methods tend to suffer from numerical instability, requiring $\mo(d\polylog (d/\eps))$ bits of precision \cite{Haah2019}. Prony’s method \cite{Ying2022} circumvents root-finding altogether, offering greater numerical stability by constructing the complementary polynomials directly. However, the stability of the layer stripping process used in most direct methods remains an open question.

Iterative approaches \cite{DongMengWhaleyEtAl2021, DongLinNiEtAl2022} typically exhibit superior numerical stability and efficiency. However, these methods face limitations in the fully-coherent regime, where $\norm{f}_{\infty}$ approaches 1. Newton's method \cite{dong2023robust} demonstrates robustness in this fully-coherent case, with a per-iteration complexity dominated by solving linear equations in $\mo(d^3)$. In the non-fully-coherent setting, this complexity can be reduced to $\mo(d^2\log\eps^{-1})$ by employing iterative solvers for matrix inversions, as the linear systems are well conditioned. These iterative methods rely heavily on the evaluation of the QSP matrix \eqref{eq: qsp matrix}, which has a computational complexity of $\mo(d^2)$. Consequently, the $\mo(d\log^2 d)$ algorithm presented in \Cref{sec: fast FPI} can significantly enhance their efficiency, especially by improving the overall asymptotic complexity of \cite{DongMengWhaleyEtAl2021, DongLinNiEtAl2022}, leading to optimal performance in the non-fully-coherent regime (up to logarithmic factors).

The most recent method, based on the nonlinear Fourier transform (NLFT), establishes a connection between QSP and NLFT, and uses the Riemann-Hilbert-Weiss (RHW) algorithm to compute the phase factors \cite{alexis2024infinite}. This is the first provably numerically stable method that performs well even in the fully-coherent regime. However, its computational complexity is quartic in $d$, which is suboptimal. In this paper, we propose a novel algorithm with quadratic complexity in $d$, making it the fastest among all known methods applicable across all regimes.

\begin{table*}[ht]
\renewcommand{\arraystretch}{1.5}
\resizebox{\textwidth}{!}{%
\begin{tabular}{cc|c|c|c}
\hline\hline
\multirow{2}{*}{}&\multirow{2}{*}{Algorithm} & Fully-coherent  & Non-fully-  & Numerical      \\ 
& & complexity & coherent &  Stability    \\ \hline
&  \cite{GilyenSuLowEtAl2018} & $\tmo(d^3\polylog\eps^{-1})^*$ & $\tmo(d^3\polylog\eps^{-1})^*$ & \xmark    \\
Direct
&  \cite{Haah2019} & $\tmo(d^3\polylog\eps^{-1})$ & $\tmo(d^3\polylog\eps^{-1})$ & \xmark   \\
method
&  \cite{ChaoDingGilyenEtAl2020} & ? & ? & \cmark$^*$    \\
& Prony \cite{Ying2022} & ? & $\tmo(d^2\polylog\eps^{-1})$ & \cmark$^*$    \\ \hline
& LBFGS \cite{DongMengWhaleyEtAl2021} & ? & $\tmo(d^2\log\eps^{-1})^*$  & \cmark   \\
Iterative
& FPI \cite{DongLinNiEtAl2022} & \xmark  & $\tmo(d^2\log\eps^{-1})$ & \cmark   \\
method
& Newton \cite{dong2023robust} & \cmark$^*$ & {\small $\tmo(d^3\log\log\eps^{-1})$ or $\tmo(d^2\log\eps^{-1})^*$} & \cmark   \\
& FFPI (Sec. \ref{sec: acc FPI})& \xmark & {\color{red}$\tmo(d\log\eps^{-1})$} & \cmark   \\ \hline
\multirow{2}{*}{NLFT} 
& RHW \cite{alexis2024infinite} & $\tmo(d^4+d\eta^{-1}\log\eps^{-1})$ & $\tmo(d^4+d\log\eps^{-1})$ & \cmark   \\
& HC (Sec. \ref{sec: fast NLFT}) & {\color{red}$\tmo(d^2+d\eta^{-1}\log\eps^{-1})$} & $\tmo(d^2+d\log\eps^{-1})$ & \cmark    \\\hline\hline
\end{tabular}
}
\caption{The comparison of different phase factor finding algorithms. The asterisk means only work for partial cases or unproved observations. In the $\tmo$ notation, we are omitting higher order $\log$ terms for conciseness, i.e. $\tmo(N) = \mo(N\polylog(N))$ for any expression $N=N(d,\eta,\eps)$. The complexity in the non-fully-coherent regime does not rely on $\eta$ since $\eta=\Omega(1)$ in this case.}
\label{table: compare}
\end{table*}

\subsection{Organization of the paper} 
The remainder of this paper is organized as follows. In \Cref{sec: fast NLFT}, we introduce the ``Half Cholesky'' algorithm. \Cref{sec: rev RHW alg} reviews some preliminaries on NLFT and the RHW algorithm, \Cref{sec: rearrange} to \ref{sec: half Chol} describe the derivation of the new algorithm, and \Cref{sec: HC complexity} to \ref{sec: HC stability} present the complexity and stability analysis. In \Cref{sec: fast FPI}, we introduce the method for accelerating the evaluation of the QSP matrix in \Cref{sec: acc qsp}, and subsequently present the "Fast Fixed Point Iteration" algorithm for retrieving phase factors in \Cref{sec: acc FPI}. Finally, numerical experiments are presented in \Cref{sec: numerics}.

\section{Fast NLFT algorithm for phase factor finding}\label{sec: fast NLFT}
\subsection{Review of NLFT and Riemann-Hilbert-Weiss algorithm}\label{sec: rev RHW alg}
We define $\DD := \{z\in \CC:\abs{z}<1\}$ as the open unit disk, $\TT := \{z\in \CC:\abs{z}=1\}$ as the unit circle, and $\DD^*:= \{z:\abs{z}>1\}\cup\{\infty\}$.

Given a compactly supported\footnote{\cite{DongLinNiEtAl2022} and \cite{alexis2023quantum} pointed out that both QSP and NLFT can be defined for more general sequences that are not compactly supported. Here we use this simpler assumption since we only discuss the algorithm aspect of the problem.} sequence $F = (F_n)_{n\in\ZZ}$ of complex numbers, we recursively define the Laurent polynomials $(a_n(z),b_n(z))$ by
\begin{equation} \label{eq: defn_NLFT_gen}
\begin{pmatrix}
   a_n (z) & b_n (z) \\
   -b_n ^* (z) & a_n ^* (z)
\end{pmatrix} = \begin{pmatrix}
   a_{n-1} (z) & b_{n-1} (z) \\
   -b_{n-1} ^* (z) & a_{n-1} ^* (z)
\end{pmatrix} \frac{1}{\sqrt{1+|F_n|^2}} \begin{pmatrix}
   1 & F_n z^n \\
   -\overline{F_n} z^{-n} & 1
\end{pmatrix}  
\end{equation}
with the initial condition
\begin{equation}
\begin{pmatrix}\label{eq: init_condition_NLFT}
    a_{-\infty} (z) & b_{-\infty} (z) \\
    -b ^* _{-\infty} (z) & a_{-\infty} ^* (z)
\end{pmatrix} = \begin{pmatrix}
    1 & 0 \\ 0 & 1
\end{pmatrix} .
\end{equation}
Here, $a^*(z) := \overline{a(\overline{z^{-1}})}$ for any function $a$. The nonlinear Fourier transform (NLFT) of the sequence $F$ is defined as the pair of Laurent polynomials 
\[
(a(z), b(z)) := (a_{\infty} (z), b_{\infty} (z)) ,
\]
where $(a_{\infty}, b_{\infty}) = (a_n, b_n)$ for all $n$ to the right of the support of $F$. Since all the matrices involved here have determinant $1$, we have
\begin{equation}\label{eq: det_cond_1}
a(z) a^* (z) + b(z) b^* (z) = 1.
\end{equation}
Conversely, given a pair $(a,b)$ satisfying \eqref{eq: det_cond_1}, we may perform the inverse NLFT of it, i.e. recover the sequence $F$, using the Riemann-Hilbert factorization.

For any $\eta > 0$, consider an even polynomial function $f\in \mathbf{S}_{\eta}$. This function can be associated with a Laurent polynomial $b(z) = \I f(x)$, where $x$ and $z$ are connected through the variable substitution 
 \begin{equation}\label{eq: change_of_variables}
         \cos(\theta) = x, \quad z= e^{2\I \theta}.
    \end{equation}
It is proved in \cite[Theorem 4]{alexis2024infinite} that we can construct a unique outer function $a(z)$ such that \eqref{eq: det_cond_1} holds. 


We denote $F$ as the inverse NLFT of this pair $(a, b)$. According to \cite[Lemma 3]{alexis2024infinite}, this $F$ is related to the QSP phase factors as follows:
\begin{equation}\label{eqn:F_psi_mapping}
    F_n = \I \tan(\psi_{\abs{n}}),\quad n\in \ZZ .
\end{equation} 

After establishing this relationship between QSP and NLFT, we can recover $\Psi$ using the following procedure. 

\textbf{Step 1:} Construct the pair $(a,b)$ from the given $f$, and calculate $\frac{b}{a}$. It turns out that $\frac{b}{a}$ can be expand as Laurent series
\begin{equation}\label{eq: c_k defi}
\frac{b(z)}{a(z)} = \sum_{k=-\infty}^{d} c_k z^k,
\end{equation}
with pure imaginary coefficients $c_k$. The coefficients $ c_0, c_1, \ldots, c_d$ can be computed using the Weiss algorithm, which can be found in \cite[Algorithm 2]{alexis2024infinite} and \cite{berntson2024complementary}. We note that in numerical computations, truncation errors may arise due to the impossibility of handling the infinite Laurent series in practice. However, these errors can be controlled by properly choosing the truncation length \cite[Theorem 8]{alexis2024infinite}.

\textbf{Step 2:} The Riemann-Hilbert factorization procedure can be translated as the following algorithm to recover $F_k$, and hence the phase factors $\psi_k$. First, construct a Hankel matrix $\Xi_k$ of size $(d+1-k)\times (d+1-k)$ with $(c_k,\ldots,c_d)^T$ as its first column and $(c_d,0,\ldots,0)$ as its last row. Then solve the linear system 
$$\begin{pmatrix}
	I &  -\Xi_k\\
	-\Xi_k & I 
\end{pmatrix}\begin{pmatrix}
	\ba_k\\
	\bb_k
\end{pmatrix} = \begin{pmatrix}
	\be_0\\
	\bzero
\end{pmatrix}$$
 for $\ba_k$ and $\bb_k$, where $\be_0$ is the first column of the identity matrix. Finally, compute $\psi_k = \arctan\left(-\I \frac{b_{k,0}}{a_{k,0}}\right)$, where $a_{k,0}$ and $b_{k,0}$ are the first entries of $\ba_k$ and $\bb_k$, respectively.
 
 Note that this algorithm requires solving a linear system for each \( \psi_k \), making the process computationally expensive when recovering the entire set of phase factors. In the following subsections, we will introduce an alternative method for this step, significantly improving the computational complexity.

\subsection{Rearrangement of the linear systems}\label{sec: rearrange}
For each integer \( k \in [0,d] \), consider the lower triangular Toeplitz matrix \( T_k \in \mathbb{C}^{(d+1-k) \times (d+1-k)} \) with the first column \( (-c_d, \ldots, -c_k)^T \). This matrix is a rearrangement of the corresponding Hankel matrix \( -\Xi_k \), so the vectors \( \ba_k \) and \( \bb_k \) satisfy the following system:
\begin{equation}\label{eq: k-th eq}
    \begin{pmatrix}
        I& T_k^T\\T_k & I
    \end{pmatrix}\begin{pmatrix}
        \bb_k\\ \rev(\ba_k)
    \end{pmatrix} = \begin{pmatrix}
        \bzero \\ \rev(\be_0)
    \end{pmatrix},
\end{equation}
where \( \rev(\bv) \) represents a vector \( \bv \) in reverse order.

Now let $T = T_0 \in \CC^{(d+1)\times(d+1)}$ with first column $(-c_{d},\ldots, -c_0)^T$, and denote $A$ as the matrix
\begin{equation}
    A:= \begin{pmatrix}
        I& T^T\\T & I
    \end{pmatrix}.
\end{equation}
If we partition \( A \) as a \( 4 \times 4 \) block matrix, with block sizes $d+1-k$, $k$, $d+1-k$, and $k$, then $A$ takes the form
\begin{equation}
    A = \begin{pmatrix}
        I_{d+1-k}&&T_k^T&*\\
        &I_k&&*\\
        T_k&&I_{d+1-k}&\\
        *&*&&I_k
    \end{pmatrix}.
\end{equation}
We may plug it into \eqref{eq: k-th eq}, and get
\begin{equation}
    A\begin{pmatrix}
        \bb_k\\ \bzero\\ \rev(\ba_k) \\ \bzero
    \end{pmatrix} = \begin{pmatrix}
        \bzero \\ \bzero \\ \rev(\be_0)\\ *
    \end{pmatrix}.
\end{equation}
By collecting the equations for all \( k \), we obtain
\begin{equation}
   A \begin{pmatrix}
       b_{d,0} & b_{d-1,0}&\cdots & b_{0,0}\\
        & b_{d-1,1}&&b_{0,1}\\
       &&\ddots&\vdots\\
       &&& b_{0,d}\\
       a_{d,0} & a_{d-1,1}&\cdots & a_{0,d}\\
        & a_{d-1,0}&&a_{0,d-1}\\
       &&\ddots&\vdots\\
       &&& a_{0,0}
   \end{pmatrix}
= \begin{pmatrix}
       &&&\\
       &&&\\
       &&&\\
       &&&\\
       1 & & & \\
       * & 1&&\\
       \vdots&&\ddots&\\
       *&*&\cdots& 1
   \end{pmatrix}
\end{equation}
We can further complement the matrices into triangular ones,
\begin{equation}\label{eq: AS=L_A explicit}
   \begin{pmatrix}
       I& T^T\\T & I
   \end{pmatrix}\begin{pNiceArray}{c|cccc}[margin=2pt]
\Block{4-1}<\Large>{I_{d+1}} 
       &b_{d,0} & b_{d-1,0}&\cdots & b_{0,0}\\
       & & b_{d-1,1}&&b_{0,1}\\
       &&&\ddots&\vdots\\
       &&&& b_{0,d}\\
       \hline
       &a_{d,0} & a_{d-1,1}&\cdots & a_{0,d}\\
       & & a_{d-1,0}&&a_{0,d-1}\\
       &&&\ddots&\vdots\\
       &&&& a_{0,0}
\end{pNiceArray} = \begin{pNiceArray}{c|cccc}[margin=2pt]
\Block{4-1}<\Large>{I_{d+1}} 
       &&&&\\
       &&&&\\
       &&&&\\
       &&&&\\
       \hline
       \Block{4-1}<\Large>{T}
       &1 & & & \\
       &* & 1&&\\
       &\vdots&&\ddots&\\
       &*&*&\cdots& 1
\end{pNiceArray}
\end{equation}
We can denote the two large matrices on both sides of this equation as \( S \) and \( L_A \), respectively. The equation can then be compactly written as:
\[
AS = L_A.
\]
This implies that \( A = L_A S^{-1} \), which gives us the LU factorization of \( A \). 

Therefore, we only need to find the LU factorization \( A = L_A U_A \) first. Then, by inverting the upper triangular factor \( U_A \), we can obtain \( S \). Finally, we can directly read off the values of \( b_{k,0}/a_{k,0} \) from \( S \).


\subsection{Exploiting the block structure of $A$}
When performing the factorization of a matrix with a block structure, it is standard practice to first apply block factorization. Besides reducing computational complexity by a constant factor, this approach offers a more significant advantage in the current context: the resulting smaller system is particularly well-suited for the fast and stable factorization algorithm introduced in the following subsections.

We begin by applying the block LDL factorization to $A$, which gives
\begin{equation}
	\begin{pmatrix}
		I& T^T\\T & I
	\end{pmatrix} = \begin{pmatrix}
		I& \\T & I
	\end{pmatrix}\begin{pmatrix}
		I& \\ & I-TT^T
	\end{pmatrix}\begin{pmatrix}
		I& T^T\\ & I
	\end{pmatrix}.
\end{equation}
Thus, the remaining task is to factorize $I-TT^T$. Defining $B = \i T$, we observe that $B$ is a real lower triangular Toeplitz matrix with its first column given by $\bp = -\i (c_0,\ldots, c_d)^T$. Since the $c_j$ coefficients are purely imaginary, $B$ is real and lower triangular. Consequently,
\begin{equation}
	I-TT^T = I+BB^T
\end{equation}
is a real symmetric positive-definite (SPD) matrix.

Denote the LDL decomposition of $I+BB^T$ as
\begin{equation}
	I+BB^T = LDL^T,
\end{equation}
which can be computed in only $\mo(d^2)$ complexity using the method detailed in \Cref{sec: half Chol}.
Therefore, it follows that the LU factorization of $A$ is
\begin{equation}
	A = L_AU_A = \left[\begin{pmatrix}
		I& \\T & L
	\end{pmatrix}\right]\left[\begin{pmatrix}
		I& \\ & DL^T
	\end{pmatrix}\begin{pmatrix}
		I& T^T\\ & I
	\end{pmatrix}\right].
\end{equation}
This implies that
\begin{equation}
	S = U_A^{-1} = \begin{pmatrix}
		I& -T^T \\ & I
	\end{pmatrix}\begin{pmatrix}
		I& \\ & L^{-T}D^{-1}
	\end{pmatrix} = \begin{pmatrix}
		I& -T^T L^{-T}D^{-1}\\ & L^{-T}D^{-1}
	\end{pmatrix}.
\end{equation}
According to \eqref{eq: AS=L_A explicit}, the relevant elements in this matrix are the diagonal elements of $L^{-T}D^{-1}$, and the first row of $-T^TL^{-T}D^{-1}$. The diagonal entries of $L^{-T}D^{-1}$ are simply the inverses of the diagonal elements of $D$, as $L$ has 1's on the diagonal. The first row of $-T^T L^{-T}D^{-1}$ can be computed by solving the triangular system $(c_d,\ldots,c_0)(DL^T)^{-1}$, which is computationally efficient. To summarize, we have
\begin{equation}
	\begin{pmatrix}
		a_{d,0}\\ \vdots \\ a_{0,0}
	\end{pmatrix} = D^{-1}\begin{pmatrix}
		1\\ \vdots \\ 1
	\end{pmatrix}
	, \text{ and } -\i \begin{pmatrix}
		b_{d,0}\\ \vdots \\ b_{0,0}
	\end{pmatrix} = D^{-1}L^{-1}\bp.
\end{equation}
where the second vector is transposed and multiplied by $-\i$. Therefore, we conclude
\begin{equation}
	\begin{pmatrix}
		\phi_{d}\\ \vdots \\ \phi_{0}
	\end{pmatrix} = \arctan \begin{pmatrix}
		-\i b_{d,0}/a_{d,0}\\ \vdots \\ -\i b_{0,0}/a_{0,0}
	\end{pmatrix}
	= \arctan(L^{-1}\bp),
\end{equation}
where the $\arctan$ function is applied componentwise.

Until now, we have completed the framework for the QSP phase factor finding algorithm, except that the fast LDL decomposition is postponed to \Cref{sec: half Chol}. We summarize the procedure in \Cref{alg: phase factor finding}. This method is termed ``half Cholesky'' (HC) because the goal is to compute $L^{-1}\bp$ rather than the full inverse of $(1+BB^T)^{-1}\bp$.

\begin{algorithm}[htbp]
	\caption{``Half Cholesky'' algorithm for phase factor finding}
	\label{alg: phase factor finding}
	\begin{algorithmic}
		\STATE{\textbf{Input:}  Target function $f$.}
		\STATE{Calculate the Laurent coefficients of $b(z) = \I f(x)$, where $x$ and $z$ are connected through \eqref{eq: change_of_variables}.}
		\STATE{Calculate the Laurent coefficients $\bc = (c_0,\ldots, c_d)^T$ of $\frac{b}{a}$ using Weiss's algorithm \cite[Algotithm 2]{alexis2024infinite}, as outlined in \Cref{sec: rev RHW alg}.}
		\STATE{Define $B$ as the real lower triangular Toeplitz matrix with first column $\bp = -\i\cdot\rev(\bc)$.}
		\STATE{Calculate the LDL factorization $I+BB^T = LDL^T$, and $\by = L^{-1}\bp$ using \Cref{alg: half cholesky}.}
		\STATE{Calculate $\Phi = \rev(\arctan(\by))$.}
		\STATE{\textbf{Output:} Reduced phase factors $\Phi$.}
	\end{algorithmic}
\end{algorithm}

\subsection{Fast LDL factorization using displacement structure}\label{sec: half Chol}
The standard Cholesky factorization of $I+BB^T$ typically incurs a complexity of $\mathcal{O}(d^3)$ since it is a dense matrix. However, we discover that $I+BB^T$ exhibits a \emph{displacement structure}, which is defined in e.g. \cite{sayed1995fast} This enables the development of an efficient LDL decomposition algorithm with complexity $\mathcal{O}(d^2)$. For simplicity, denote $K = I+BB^T$. 
By direct calculation, it can be verified that
\begin{equation}\label{eq: displacement K}
    K - ZKZ^T = GG^T,
\end{equation}
where $Z$ is the lower shift matrix, and $G = [\be_0,\bp]\in \RR^{(d+1)\times 2}$ with $\be_0$ being the first column of the identity matrix, and $\bp$ the first column of $B$. This is a special case of the displacement structure. It has long been established that matrices with displacement structure can be factorized in $\mathcal{O}(d^2)$ time \cite{sayed1995fast} using Schur's algorithm or its variants. For completeness, we briefly outline the algorithm for this specific structure \eqref{eq: displacement K}. 

The key idea is to recover only the first column of $K$ from \eqref{eq: displacement K} based on $G$, instead of explicitly constructing the entire matrix $K$. This is sufficient to construct the first column of the $L$ factor. Moreover, we can identify a similar displacement structure for the first Schur complement of $K$, allowing the process to be applied recursively.

We begin by performing the right QR decomposition of $G$, such that $G = RQ$, where $Q \in SO(2)$ is orthogonal, and the first row of $R$ takes the form $[*, 0]$. From this, we also have a similar displacement equation
\begin{equation}\label{eq: displacement with R}
    K - ZKZ^T = RR^T,
\end{equation}
We may write $R$ explicitly as
\begin{equation}
    R = \begin{pmatrix}
        s & 0\\
        \bu &\bv
    \end{pmatrix},
\end{equation}
Consequently, the first column of $K$ is given by $\begin{pmatrix}
	s^2\\ s\bu
\end{pmatrix}$, and thus the first step of LDL decomposition of $K$ is
\begin{equation}
    K = \begin{pmatrix}
        1 & 0\\
        \bu/s & I
    \end{pmatrix}\begin{pmatrix}
        s^2 & 0\\
        0 & K_1
    \end{pmatrix}\begin{pmatrix}
        1 & \bu^T/s\\
        0 & I
    \end{pmatrix} = \begin{pmatrix}
        s^2&s\bu^T\\s\bu&\bu\bu^T+K_1
    \end{pmatrix},
\end{equation}
where $K_1$ is the first Schur complement of $K$. Substitute this expression into \eqref{eq: displacement with R} and we obtain
$$\begin{pmatrix}
        s^2&s\bu^T\\s\bu&\bu\bu^T+K_1
    \end{pmatrix} - \begin{pmatrix}
        0&\\ \be_0& Z_d
    \end{pmatrix}\begin{pmatrix}
        s^2&s\bu^T\\s\bu&\bu\bu^T+K_1
    \end{pmatrix}\begin{pmatrix}
        0& \be_0^T\\& Z_d^T
    \end{pmatrix}=RR^T = \begin{pmatrix}
        s^2&s\bu^T\\s\bu&\bu\bu^T+\bv\bv^T
    \end{pmatrix},$$
where $Z_d$ is the lower shift matrix of dimension $d$. By comparing the lower-right block, we deduce that $K$'s first Schur complement $K_1$ satisfies
\begin{equation}\label{eq: K2 displacement}
    K_1-Z_d K_1 Z_d^T = \bw\bw^T + \bv\bv^T,
\end{equation}
where $\bw = s\be_0+Z\bu$ is the vector obtained by deleting the last element of the first column of $R$. If denote $G_1 = [\bw, \bv]$, then \eqref{eq: K2 displacement} suggest $K_1$ has the same displacement structure as $K$. Therefore, we can repeat this process and get a sequence of rank 2 matrices $G_0=G, G_1, \ldots, G_d$, and calculate the columns of the LDL decomposition of $K$, without explicitly constructing $K$ or its Schur complements. Once $L$ is constructed, solving the triangular system for $\by = L^{-1}\bp$ becomes straightforward.

The full procedure is summarized in \Cref{alg: half cholesky}. To present the algorithm more compactly, we do not reduce the dimensions of $G_k$ and $R_k$ at each step. Instead, we pad zeros at the top to maintain their dimensions as $(d+1)\times 2$. In particular, we have the correspondence $\bu_0=\begin{pmatrix}
    s\\ \bu
\end{pmatrix}$ and $\bv_0 = \begin{pmatrix}
    0\\ \bv
\end{pmatrix}$.

\begin{algorithm}[htbp]
	\caption{``Half Cholesky'' solver for $K:=I+BB^T$ based on Schur's algorithm.}
	\label{alg: half cholesky}
	\begin{algorithmic}
		\STATE{\textbf{Input:}  Column vector $\bp\in\RR^{d+1}$, which is the first column of $B$.}
		\STATE{Construct $G_0 = [\be_0, \bp]$.}
		\STATE{Initialize $L = [\bl_0, \ldots, \bl_d] = 0\in\RR^{(d+1)\times (d+1)}.$}
		\FOR{$k = 0,1,\ldots, d$}
		\STATE{Calculate the right QR decomposition $G_k = R_kQ_k$, where $R_k = [\bu_k, \bv_k]\in\RR^{(d+1)\times 2}$. (Thus the first $k$ elements of $\bu_k$ and the first $k+1$ elements of $\bv_k$ are all $0$'s.)}
		\STATE{Update the $k$-th column of $L$ as $\bl_{k} = \frac{1}{u_{k,k}}\bu_k$.}
		\STATE{$G_{k+1} = \begin{bmatrix}
				Z\bu_k, \bv_k
			\end{bmatrix}\in \RR^{(d+1)\times 2}$.}
		\ENDFOR
		\STATE{Calculate $\by = L^{-1}\bp$.}
		\STATE{\textbf{Output:} $\by$.}
	\end{algorithmic}
\end{algorithm}

\subsection{Complexity Analysis}\label{sec: HC complexity}

\Cref{alg: phase factor finding} first computes the coefficients \( \bc \) using the Weiss algorithm, followed by the Cholesky factorization as outlined in \Cref{alg: half cholesky}. As previously discussed, \Cref{alg: half cholesky} has a computational complexity of \( \mo(d^2) \). According to \cite[Theorem 8]{alexis2024infinite}, the Weiss algorithm involves performing several fast Fourier transforms on vectors of length \( N = \mo\left(\frac{d}{\eta}\log\left(\frac{d}{\eta\epsilon}\right)\right) \), resulting in a complexity of \( \mo(N\log N) = \tmo(N) = \tmo\left(\frac{d}{\eta}\log\left(\frac{d}{\eta\epsilon}\right)\right) \). Therefore, we conclude that the total computational complexity of \Cref{alg: phase factor finding} is 
$$\tmo\left(d^2 + \frac{d}{\eta}\log\left(\frac{d}{\eta\epsilon}\right)\right).$$

\subsection{Stability Analysis}\label{sec: HC stability}

Assume that the machine precision is $\ema$. We define an algorithm as stable if the final round-off error $\eps$ satisfies $\eps = \mo(\poly(d,\eta^{-1}))\cdot\ema$. This implies that the bit requirement $\log(1/\ema) = \mo(\log(d\eta/\eps))$ scales logarithmically with all parameters, ensuring the algorithm generally performs well on double precision (64-bit) computers.

The main components of \Cref{alg: phase factor finding} and \Cref{alg: half cholesky} are the Fast Fourier Transform (FFT) used in Weiss's algorithm, the triangular system solver, and the LDL decomposition with displacement structure. The stability of the first two components follows standard numerical algebra results \cite{alexis2024infinite, higham2002accuracy, ramos1971roundoff}. Thus, our focus is on the stability of the LDL decomposition with displacement structure.

Many fast solvers for linear systems with displacement structures are unstable. \cite{kailath1999fast} 
However, stability results can be established for certain matrix classes, such as Toeplitz-like SPD matrices \cite{chandrasekaran1996stabilizing, stewart1997stability}. While no existing analysis directly addresses our specific case \eqref{eq: displacement K}, we propose a similar approach for analysis. In the following theorem, we establish a stability result for a fast solver applied to matrices $K$ satisfying \eqref{eq: displacement K} for arbitrary $G$.

\begin{theorem}\label{thm: stability}
	Assume $K$ satisfy \eqref{eq: displacement K} for some $G\in \RR^{(d+1)\times 2}$. The matrix $\hL$ is the matrix we obtain after executing \Cref{alg: half cholesky} using floating point arithmetic of machine precision $\ema$, then we have backward error estimation in Frobenius norm
	\begin{equation}\label{eq: stab backward}
		\norm{\hL\hh^2\hL^T-K}_F\le \mo(d^3)\norm{G}_F^2\ema
	\end{equation}
	for some diagonal matrix $\hh$. We also have the forward error estimation
	\begin{equation}
		\norm{\hat{L}-L}_F\le \mo(d^3\eta^{-1})\norm{G}_F^2\ema.
	\end{equation}
\end{theorem}

\begin{proof}
We introduce a matrix $U := [\bu_0,\ldots,\bu_d]$, and $H := \diag\{u_{00},\ldots, u_{dd}\}$, so that $LH = U$ and $UU^T = K$. We denote $\hat{\cdot}$ as the corresponding values with round-off errors calculated using the floating point arithmetic. In particular, the $\hh$ in \eqref{eq: stab backward} is $\hh=\diag\{\hat{u}_{00},\ldots,\hat{u}_{dd}\}$. 
In the following proof, we always assume that $\mathrm{poly}(d)\ema \ll 1$, as is the standard assumption in stability analysis. We use $C$ with subscripts to represent constants. 

According to the error analysis of QR factorization based on Householder transformations \cite[Theorem 19.4]{higham2002accuracy}, there exists some orthogonal matrix $\tilde{Q}_k$\footnote{We emphasize that this $\tilde{Q}_k$ is an exact orthogonal matrix, and it may differ slightly from the computed $\hq_k$.} such that
\begin{equation}\label{eq: E_k}
    \hr_k\tilde{Q}_k^T = \hg_k + E_k
\end{equation}
with $\norm{E_k}_F\le C_1\ema d\norm{\hg_k}_F$. This means $\norm{\hg_{k+1}}_F\le \norm{\hr_k}_F\le (1+C_1\ema d)\norm{\hg_k}_F$, and thus
\begin{equation}\label{eq: hat G_k}
    \norm{\hg_k}_F\le (1+C_1\ema d)^k \norm{G_0}_F \le C_2 \norm{G_0}_F.
\end{equation}

Multiplying \eqref{eq: E_k} by its transpose, we know that there exists some symmetric matrix $M_k$ such that 
\begin{equation}\label{eq: M_k}
    \hg_k\hg_k^T+M_k = \hr_k\hr_k^T = \hu_k\hu_k^T + \hv_k\hv_k^T,
\end{equation}
with $\norm{M_k}_F\le C_3\ema d\norm{\hg_k}_F^2\le C_4\ema d\norm{G_0}_F^2$.

By definition, we have
$$Z\hu_k\hu_k^T Z^T + \hv_k\hv_k^T = \hg_{k+1}\hg_{k+1}^T.$$
Combined with \eqref{eq: M_k}, we get
\begin{equation}\label{eq: rec}
    Z\hu_k\hu_k^T Z^T + \hv_k\hv_k^T = \hu_{k+1}\hu_{k+1}^T + \hv_{k+1}\hv_{k+1}^T - M_{k+1}.
\end{equation}
Summing \eqref{eq: rec} over $k$, and we may obtain 
$$G_0G_0^T + \sum_{k=0}^d M_k = \hU\hU^T -Z\hU\hU^TZ^T.$$
After subtracting the displacement equation $G_0G_0^T = K-ZKZ^T$, we get
\begin{equation}\label{eq: diff UU^T-K}
    (\hU\hU^T-K) - Z(\hU\hU^T-K)Z^T = \sum_{k=0}^d M_k.
\end{equation}

Define an operator $\mathcal{A}$ on the space of matrices as 
$$\mathcal{A}: X \mapsto X - ZXZ^T.$$
If we flatten the $(d+1)\times (d+1)$ matrix $X$ as an vector of length $(d+1)^2$, then this operator has a matrix form $\mathcal{A} = I - Z\otimes Z$. The Frobenius norm of $X$ is the same as the 2 norm of the flattened vector, so the induced operator norm of $\mathcal{A}^{-1}$ is $\norm{\mathcal{A}^{-1}}_{\mathrm{op}} = \norm{(I-Z\otimes Z)^{-1}}_2 = \mo(d)$. 

Now \eqref{eq: diff UU^T-K} is equivalent to $\mathcal{A}(\hU\hU^T-K) = \sum_{k=0}^d M_k$, therefore
\begin{equation}\label{eq: UU-K est}
    \norm{\hU\hU^T-K}_F \le \norm{\mathcal{A}^{-1}}_{\mathrm{op}}\norm{\sum_{k=0}^d M_k}_F\le C_5\ema d^3\norm{G_0}_F^2.
\end{equation}
Next, we use \eqref{eq: hat G_k} to estimate $\norm{\hU}_F$ as
\begin{equation}\label{eq: hu est}
    \norm{\hU}_F = \sqrt{\sum_{k=0}^d\norm{\hu_k}^2} \le \sqrt{\sum_{k=0}^d\norm{\hg_k}_F^2}\le C_2\sqrt{d+1}\norm{G_0}_F,
\end{equation}
Recall that $L=UH^{-1}$, and $\hL$ is obtained by dividing the $\hU$ by its diagonal component $\hh=\diag\{\hat{u}_{00},\ldots,\hat{u}_{dd}\}$, so we have
\begin{equation}\label{eq: U-LH est}
    \norm{\hU-\hL\hh}_F\le C_6\ema \norm{\hU}_F = C_2C_6\ema\sqrt{d+1}\norm{G_0}_F,
\end{equation}
where the last step uses \eqref{eq: hat G_k}. By exploiting \eqref{eq: hu est} and \eqref{eq: U-LH est}, we obtain
\begin{align*}
    \norm{\hU\hU^T-\hL\hh^2\hL^T}_F  &=\norm{(\hU-\hL\hh)\hU^T-(\hU-\hL\hh)(\hU-\hL\hh)^T + \hU(\hU-\hL\hh)^T}_F\\
    &\le 2\norm{\hU}_F\norm{\hU-\hL\hh}_F + \norm{\hU-\hL\hh}_F^2\le C_7d\ema\norm{G_0}_F^2
\end{align*}
This formula together with \eqref{eq: UU-K est} leads to 
$$\norm{\hL\hh^2\hL^T-K}_F\le \norm{\hL\hh^2\hL^T-\hU\hU^T}_F + \norm{\hU\hU^T-K}_F \le C_8\ema d^3\norm{G_0}_F^2.$$
Notice that $L\cdot (H^2L^T)$ is the LU factorization of $K$, we may use the sensitive analysis of LU factorization \cite[Theorem 9.15]{higham2002accuracy} to deduce that
\begin{equation}
\begin{aligned}
    \norm{\hat{L}-L}_F&\le C_9\norm{L}_2\norm{L^{-1}}_2\norm{(H^2L^T)^{-1}}_2\frac{\norm{K}_2}{\norm{K}_F}\norm{\hL\hh^2\hL^T-K}_F\\
    &\le C_{10} \eta^{-1/2}\cdot \eta^{-1/2}\cdot 1\cdot 1\cdot \ema d^3\norm{G_0}_F^2 = C_{10} \ema d^3\eta^{-1}\norm{G_0}_F^2.
\end{aligned}
\end{equation}
Here we used \Cref{lem: norm est} in the second inequality. 
\end{proof}


\begin{lemma}\label{lem: norm est}We have the following estimations
    $$\norm{L^{-1}}_2\le \mo(\eta^{-1/2}),$$
    $$\norm{L}_2\le \mo(\eta^{-1/2}),$$
    $$\norm{(H^2L^T)^{-1}}_2\le \mo(1).$$
\end{lemma}
\begin{proof}
	Let $Y$ be the permutation matrix with $1$'s on the counter-diagonal.
	By definition, we have $B = \I T = -\I Y\Xi_0$, where $\Xi_0$ is defined in \Cref{sec: rev RHW alg}. Using this relationship, we may deduce $K = I+BB^T = Y(1-\Xi_0\Xi_0^T)Y = Y(1-(\Xi_0)^2)Y$. Therefore, we may use the estimation in \cite{alexis2024infinite} that
	$$\norm{K}_2 = \left\|I-\left(\Xi_0\right)^2\right\|_2 \leq 1+\left\|\frac{b}{a}\right\|_{L^{\infty}(\TT)}^2,$$
	By \eqref{eq: det_cond_1} and the assumption $\norm{b}_{L^{\infty}(\TT)}\le 1-\eta$, we have
	\begin{equation}\label{eq: b/a inf norm}
		1+\left\|\frac{b}{a}\right\|_{L^{\infty}(\TT)}^2\le 1+\frac{(1-\eta)^2}{1-(1-\eta)^2}\le \eta^{-1}.
	\end{equation}
	Since $K$ is SPD, its eigenvalues and singular values coincide, so we have
	$$1\le \lambda_{\min}(K)\le \lambda_{\max}(K)\le \eta^{-1}.$$
	By the relationship $K = UU^T$, we get its singular value bounds
	$$1\le \sigma_{\min}(U)\le \sigma_{\max}(U)\le \eta^{-1/2}.$$
	Since $H$ is the diagonal part of $U$, we have
	$$\norm{H}_2=\max_k\{u_{kk}\}\le \sigma_{\max}(U)\le \eta^{-1/2}.$$
	As $U$ is lower diagonal, $H^{-1}$ is also the diagonal part of $U^{-1}$, so similarly
	$$\norm{H^{-1}}_2=\max_k\{u_{kk}^{-1}\}\le \sigma_{\max}(U^{-1})\le 1.$$
	Therefore, we have the estimations
	$$\norm{L}_2 = \norm{UH^{-1}}_2 \le \norm{U}_2\norm{H^{-1}}_2 \le \eta^{-1/2}\cdot 1 = \eta^{-1/2},$$
	$$\norm{L^{-1}}_2 = \norm{HU^{-1}}_2 \le \norm{H}_2\norm{U^{-1}}_2 \le \eta^{-1/2}\cdot 1 = \eta^{-1/2},$$
	$$\norm{(H^2L^T)^{-1}}_2 = \norm{U^{-T}H^{-1}}_2 \le \norm{U^{-T}}_2\norm{H^{-1}}_2 \le 1\cdot 1 = 1.$$	
\end{proof}

Now we apply \Cref{thm: stability} to our specific case $K = I + BB^T$, where we have $G = [\be_0, \bp]$. By \eqref{eq: c_k defi}, the vector $\bc = (c_0,\ldots, c_d)^T$ is the partial Fourier coefficient of the function $\frac{b}{a}$.. By Parseval's theorem, we have $\norm{\bc}_2\le \left\|\frac{b}{a}\right\|_{L^{\infty}(\TT)}$. Therefore, it follows from \eqref{eq: b/a inf norm} that
$$\norm{G}_F^2 = \norm{\bp}_2^2+1=\norm{\bc}_2^2+1 \le \left\|\frac{b}{a}\right\|_{L^{\infty}(\TT)}^2+1\le \eta^{-1}.$$
Together with \Cref{thm: stability}, we get the stability result
$$\norm{\hat{L}-L}_F\le \eps = \mo(\poly(d,\eta^{-1}))\cdot\ema.$$

\section{Efficient evaluation of QSP matrix and Fast FPI method}\label{sec: fast FPI}
\subsection{Efficient evaluation of QSP matrix}\label{sec: acc qsp}
An important subroutine in QSP algorithms is the evaluation of the QSP matrix \eqref{eq: qsp matrix}, especially the imaginary part of its upper left entry $g(x,\Phi)$. This evaluation is essential in the iterative methods used in QSP algorithms \cite{DongMengWhaleyEtAl2021, DongLinNiEtAl2022,dong2023robust} and is critical for verifying whether an algorithm generates the correct phase factors.

In this section, we do not restrict $\Phi$ to be symmetric, but still assume $n=2d$ is even for convenience. The odd \( n \) case can be handled by removing the last phase factor and performing an additional matrix multiplication after running the algorithm for even \( n \). In the context of QSP, there are two ways to represent the even polynomial $g(x,\Phi)$. The first representation is through its values on the Chebyshev grid
$$x_j= \cos\left(\frac{2\pi j}{4d+1}\right),\ j = 0,\ldots, d,$$
and the second is via its Chebyshev coefficients $(q_0, q_1, \dots, q_d)$, where
\begin{equation}\label{eq: cheb coef}
	g(x,\Phi) = \sum_{j=0}^{d} q_j T_{2j}(x),
\end{equation}
in which $T_{2j}(x) = \cos(2j\arccos(x))$ is the $(2j)$-th Chebyshev polynomial. These two representations can be converted to each other conveniently using FFT. A direct evaluation of $g(x,\Phi)$ by the definition \eqref{eq: qsp matrix} has complexity $\mathcal{O}(d^2)$, as each $g(x_j,\Phi)$ requires $\mathcal{O}(d)$ time. We aim to reduce this complexity to $\mathcal{O}(d \log^2 d)$.

We begin by introducing the substitution $t = e^{\I \theta} = x + \I \sqrt{1 - x^2}$, transforming the QSP matrix into
\begin{equation}\label{eq: UPQ}
	U(x,\Phi)=  \left(\prod_{j=0}^{n-1} \left[ e^{\I \phi_j \sigma_z} W(x) \right]\right) e^{\I \phi_n \sigma_z}=\left(\prod_{j=0}^{n-1} V_j(t)\right) e^{\I \phi_n \sigma_z},
\end{equation}
where
\begin{equation}
	V_j(t) = \begin{pmatrix}
		e^{\i\phi_j}\frac{t+t^{-1}}{2} & e^{\i\phi_j}\frac{t-t^{-1}}{2} \\
		e^{-\i\phi_j}\frac{t-t^{-1}}{2} & e^{-\i\phi_j}\frac{t+t^{-1}}{2} 
	\end{pmatrix}.
\end{equation}
Thus, computing $U(x,\Phi)$ reduces to finding the product of $n$ matrices, where each matrix entry is a Laurent series of degree 1. Consequently, the resulting matrix $U( \frac{t + t^{-1}}{2}, \Phi ) = U(x,\Phi)$ will have entries that are Laurent series of degree $n$.

Next, we review a fast algorithm that calculates the product of a large number of polynomials. It is well known that the product of two degree-$n$ polynomials can be computed in $\mo(n\log n)$ time using FFT. (Here the product of polynomials means given the coefficients of each polynomial and try to compute the coefficients of their product.) When trying to compute the product of $n$ polynomials $p_1,p_2,\ldots$ with degree 1, we may use the divide and conquer method to leverage this fast polynomial multiplying method \cite[Chapter 4.1]{blahut2010fast}. We may first compute the adjacent products $p_{2k-1}p_{2k}$. Denote $p_{12}:=p_1p_2$, $p_{34} = p_3p_4$, etc. Then compute the product of every group of 4, such as $p_{1234}:=p_{12}p_{34}$, $p_{5678}:= p_{56}p_{78}$, etc. All these polynomial products use FFT. Then after $\mo(\log_2 n)$ rounds, we obtain the coefficients of the product of all $n$ polynomials. It can be verified that the complexity of each round is $\mo(n\log n)$, and hence the total complexity is $\mo(n\log ^2 n)$.

A similar approach applies to the evaluation of the QSP matrix. The matrix elements of the $V_j(t)$'s are Laurent series, so their products can also be computed efficiently using FFT. To calculate the product of the $d$ unitaries $V_j$, we use the same divide-and-conquer approach as with polynomial multiplication: first, compute the product of adjacent pairs, then groups of four, and so on. The intermediate products take the form
\begin{equation}
\prod_{j=m_1}^{m_2-1} V_j(t) = \begin{pmatrix}
    P(t)&Q(t)\\ \overline{Q}(t) & \overline{P}(t)
\end{pmatrix},
\end{equation}
where $P$ and $Q$ represent finite Laurent series of degree $m_2-m_1$, and $\overline{P}$ is the Laurent series obtained by taking conjugate of the coefficient of $P$ for each term. Therefore, when we multiply two matrices of this form, we have
\begin{equation}
\begin{pmatrix}
    P_1(t)&Q_1(t)\\ \overline{Q_1}(t) & \overline{P_1}(t)
\end{pmatrix}\begin{pmatrix}
    P_2(t)&Q_2(t)\\ \overline{Q_2}(t) & \overline{P_2}(t)
\end{pmatrix} = \begin{pmatrix}
    P_{12}(t)&Q_{12}(t)\\ \overline{Q_{12}}(t) & \overline{P_{12}}(t)
\end{pmatrix},
\end{equation}
where $P_{12}(t) = P_1(t)P_2(t)+Q_1(t)\overline{Q_2}(t)$ and $Q_{12}(t) = P_1(t)Q_2(t)+Q_1(t)\overline{P_2}(t)$. These operations involve standard polynomial products and sums, which are efficiently computed using FFT. After constructing the upper left corner of $U(\frac{t+t^{-1}}{2},\Phi) = U(x,\Phi)$ as a Laurant series of the form
$$p(x) = p\left(\frac{t+t^{-1}}{2}\right) = \sum_{j=0}^d h_j\cdot\frac{t^{2j}+t^{-2j}}{2},$$
the Chebyshev coefficients in \eqref{eq: cheb coef} are $q_j = \operatorname{Im}(h_j)$.

\subsection{Application to FPI method: FFPI method}\label{sec: acc FPI}
The FPI algorithm for determining phase factors is introduced in \cite{DongLinNiEtAl2022}, which heavily relies on the evaluation of the QSP matrix. We first introduce a map $F:\RR^{d+1}\to \RR^{d+1}:$
$$\Psi\mapsto\Phi\mapsto \bq = (q_0,\ldots, q_d),$$
where $\Psi$ is a sequence of reduced phase factors, and $\bq$ is defined by \eqref{eq: qsp matrix} as the Chebyshev coefficients of the corresponding full phase factor sequence $\Phi$.\footnote{We adopt different notations from \cite{DongLinNiEtAl2022}, such as a switch in the roles of $\Phi$ and $\Psi$. The $d$ and $\tilde{d}-1$ there means $n$ and $d$ here, and we also drop the $\half$ in definition the first factor $\psi_0$. However, we maintain self-consistent notation throughout this paper.}

Let $\hat{f}=(\hat{f}_0,\ldots,\hat{f}_{d})\in\RR^{d+1}$ denote the Chebyshev coefficients of the objective function $f$, that is
$$f(x) = \sum_{j=0}^d \hat{f}_j T_{2j}(x).$$
Then the QSP problem can be restated as solving a nonlinear equation $F(\Psi) = \hat{f}$. The FPI algorithm is given by the fixed point iteration
\begin{equation}\label{eq: FPI iter}
	\Psi^{(0)} = 0,\quad \Psi^{(k+1)} = \Psi^{(k)} - \half \left(F(\Psi^{(k)}) - \hat{f}\right).
\end{equation}
This iteration relies primarily on the calculation of $F$, which is based on the evaluation of the QSP matrix. It is proved in \cite{DongLinNiEtAl2022} that whenever $\|\hat{f}\|_1\le 0.861$, the FPI method will converge to the desired solution in $\mo(\log\frac{1}{\eps})$ iterations. With our enhancement of the \( F \) evaluation from \( \mathcal{O}(d^2) \) to \( \mathcal{O}(d \log^2 d) \), the overall complexity of the algorithm is consequently reduced to \( \mathcal{O}\left( d \log^2 d \log \frac{1}{\epsilon} \right) \). We name this method as the Fast Fixed Point Iteration method (FFPI).

\section{Numerical experiments}\label{sec: numerics}

We evaluate the performance of the Half Cholesky (HC) method and the Fast Fixed Point Iteration (FFPI) method in the numerical tests presented. All experiments were conducted using MATLAB R2020a on a computer with an Intel Core i5 Quad CPU running at 2.11 GHz and 8 GB of RAM. The error threshold was set to \( \epsilon = 10^{-12} \) for all experiments. The codes are open-sourced in \url{https://github.com/HongkangNi/Fast-Phase-Factor-Finding-for-QSP}.

The first test involves a random polynomial \( f \) with \( \|f\|_{\infty} = 0.5 \). This polynomial was generated by randomly selecting the Chebyshev coefficients of \( f \) and then rescaling it according to its infinity norm. This is a non-fully-coherent example, meaning that most existing methods can successfully solve the QSP problem. We compare the two new methods (HC and FFPI) with three of the fastest existing methods: Newton's method \cite{dong2023robust}, the Fixed Point Iteration (FPI) method \cite{DongLinNiEtAl2022}, and the Riemann-Hilbert-Weiss algorithm \cite{alexis2024infinite}. For comparisons with other, slower methods, we refer readers to \cite{DongMengWhaleyEtAl2021, DongLinNiEtAl2022, dong2023robust}. The runtime comparison is shown in \Cref{fig: test 1}, where it can be seen that HC and FFPI significantly outperform the existing methods in terms of efficiency. Notably, the runtime of FFPI grows almost linearly with \( d \), making it the most efficient method in the non-fully-coherent setting. This observation is consistent with the theoretical analysis.

\begin{figure}[H]
	\centering
	\includegraphics[
	width =0.5\textwidth
	]{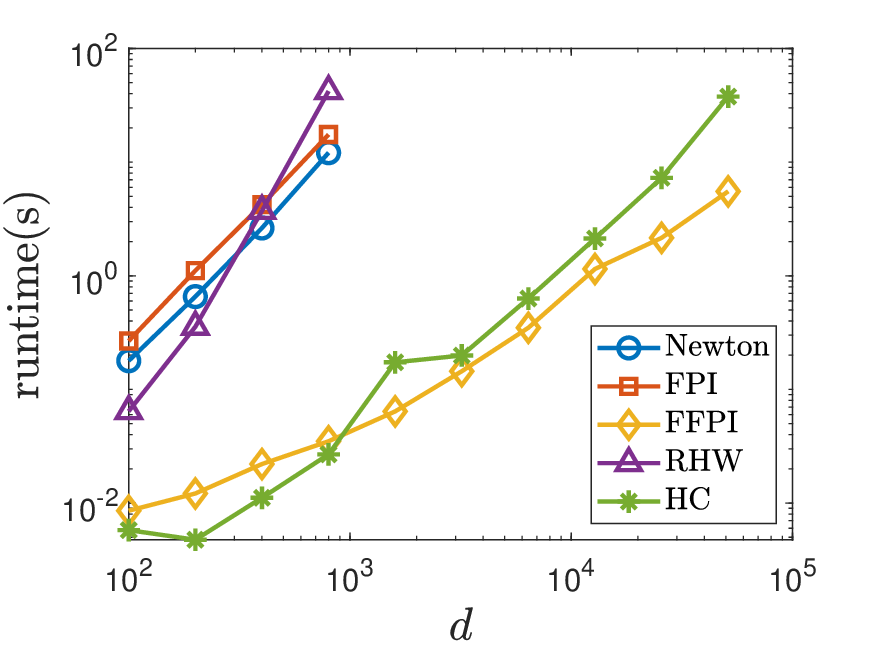}
	\caption{Example for random generated target polynomial $f$ with $\norm{f}_{\infty}=0.5$. This figure shows the runtime comparison among the five methods: Newton, FPI, FFPI, Riemann-Hilbert-Weiss, and Half Cholesky method.}
	\label{fig: test 1}
\end{figure}

The second example involves the even function \( f_{\mathrm{HS}}(x) = \cos(\tau x) \), where \( \tau \) is a parameter. This function is essential for Hamiltonian simulation on quantum computers, and more details can be found in \cite{GilyenSuLowEtAl2019, DongLinNiEtAl2022}. Since \( f_{\mathrm{HS}} \) is not a polynomial, we approximate it using its Taylor series expansion, known as the Jacobi-Anger expansion:
\begin{equation}
	\label{eq;Jacobi-Anger}
	\cos(\tau x) = J_0(\tau) + 2\sum_{k \text{ even}} (-1)^{k/2} J_k(\tau) T_k(x),
\end{equation}
where \( J_k \) are Bessel functions of the first kind. To achieve an approximation accuracy of machine precision \( \epsilon_0 = 10^{-15} \), we truncate the series after \( n = 1.4|\tau| + \log(1/\epsilon_0) \) terms. To demonstrate that the FFPI method performs well in the fully-coherent regime and to ensure that the \( L^{\infty} \) norm of \( f \) is less than 1, we define the function:
\begin{equation}
	f(x) := 0.999\left(J_0(\tau) + 2\sum_{k \text{ even}, k<n} (-1)^{k/2} J_k(\tau) T_k(x)\right).
\end{equation}
This corresponds to \( \eta = 10^{-3} \). The best previous methods for the fully-coherent regime are Newton's method and the Riemann-Hilbert-Weiss algorithm, and we compare them with the newly developed Half Cholesky method in \Cref{fig: test 2}. The FPI and FFPI methods generally fail to converge in this fully-coherent setting. To observe the runtime scaling with respect to the polynomial degree, we note that \( d = \frac{n}{2} \approx 0.7\tau \). The results show that the Half Cholesky method significantly outperforms the other methods, demonstrating its robustness and efficiency across all regimes for the QSP problem.

\begin{figure}[H]
	\centering
	\includegraphics[width=0.5\textwidth]{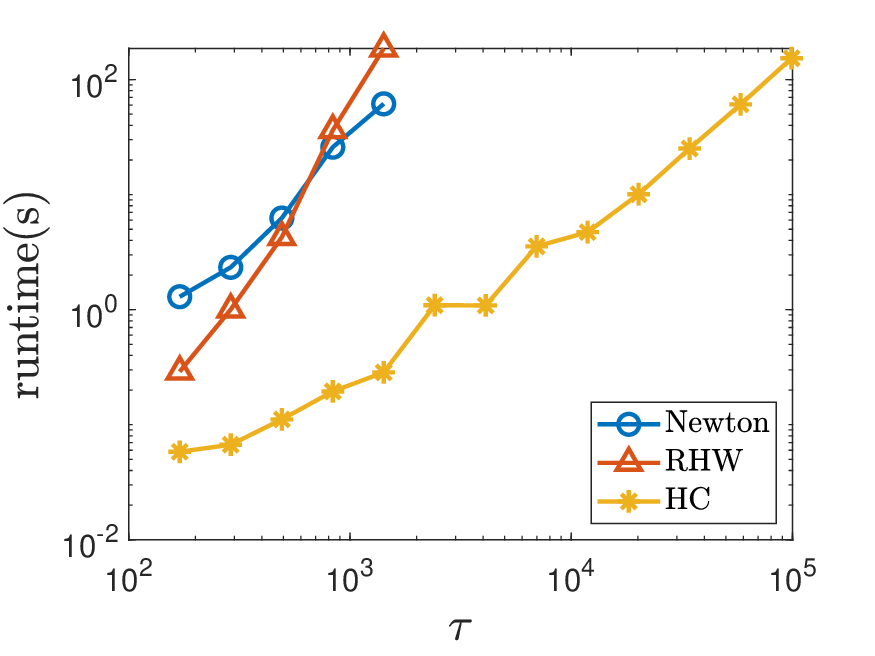}
	\caption{Example for Hamiltonian simulation, with $\norm{f}_{\infty} = 0.999$. This figure shows the runtime comparison for the fully-coherent regime among the three methods: Newton, Riemann-Hilbert-Weiss, and Half Cholesky method.}
	\label{fig: test 2}
\end{figure}

\bibliographystyle{amsalpha}
\bibliography{ref}	

\newcommand{\etalchar}[1]{$^{#1}$}
\providecommand{\noopsort}[1]{}\providecommand{\singleletter}[1]{#1}%
\providecommand{\bysame}{\leavevmode\hbox to3em{\hrulefill}\thinspace}
\providecommand{\MR}{\relax\ifhmode\unskip\space\fi MR }
\providecommand{\MRhref}[2]{%
  \href{http://www.ams.org/mathscinet-getitem?mr=#1}{#2}
}
\providecommand{\href}[2]{#2}
\begin{thebibliography}{DMWL21}

\bibitem[ALM{\etalchar{+}}24]{alexis2024infinite}
Michel Alexis, Lin Lin, Gevorg Mnatsakanyan, Christoph Thiele, and Jiasu Wang, \emph{Infinite quantum signal processing for arbitrary szegő functions}, arXiv preprint arXiv:2407.05634 (2024).

\bibitem[AMT23]{alexis2023quantum}
Michel Alexis, Gevorg Mnatsakanyan, and Christoph Thiele, \emph{Quantum signal processing and nonlinear fourier analysis}, arXiv preprint arXiv:2310.12683 (2023).

\bibitem[Bla10]{blahut2010fast}
Richard~E Blahut, \emph{Fast algorithms for signal processing}, Cambridge University Press, 2010.

\bibitem[BS24]{berntson2024complementary}
Bjorn~K Berntson and Christoph S{\"u}nderhauf, \emph{Complementary polynomials in quantum signal processing}, arXiv preprint arXiv:2406.04246 (2024).

\bibitem[CDG{\etalchar{+}}20]{ChaoDingGilyenEtAl2020}
Rui Chao, Dawei Ding, Andras Gilyen, Cupjin Huang, and Mario Szegedy, \emph{Finding angles for quantum signal processing with machine precision}, arXiv preprint arXiv:2003.02831 (2020).

\bibitem[CS96]{chandrasekaran1996stabilizing}
Shivkumar Chandrasekaran and Ali~H Sayed, \emph{Stabilizing the generalized schur algorithm}, SIAM Journal on Matrix Analysis and Applications \textbf{17} (1996), no.~4, 950--983.

\bibitem[DL21]{DongLin2021}
Yulong Dong and Lin Lin, \emph{Random circuit block-encoded matrix and a proposal of quantum linpack benchmark}, Phys. Rev. A \textbf{103} (2021), no.~6, 062412.

\bibitem[DLNW22]{DongLinNiEtAl2022}
Yulong Dong, Lin Lin, Hongkang Ni, and Jiasu Wang, \emph{Infinite quantum signal processing}, arXiv preprint arXiv:2209.10162 (2022).

\bibitem[DLNW23]{dong2023robust}
\bysame, \emph{Robust iterative method for symmetric quantum signal processing in all parameter regimes}, arXiv preprint arXiv:2307.12468 (2023).

\bibitem[DLT22]{DongLinTong2022}
Yulong Dong, Lin Lin, and Yu~Tong, \emph{Ground-state preparation and energy estimation on early fault-tolerant quantum computers via quantum eigenvalue transformation of unitary matrices}, PRX Quantum \textbf{3} (2022), 040305.

\bibitem[DMWL21]{DongMengWhaleyEtAl2021}
Yulong Dong, Xiang Meng, K~Birgitta Whaley, and Lin Lin, \emph{Efficient phase factor evaluation in quantum signal processing}, Phys. Rev. A \textbf{103} (2021), 042419.

\bibitem[DWL22]{DongWhaleyLin2021}
Yulong Dong, K~Birgitta Whaley, and Lin Lin, \emph{A quantum hamiltonian simulation benchmark}, npj Quantum Information \textbf{8} (2022), no.~1, 131.

\bibitem[GSLW18]{GilyenSuLowEtAl2018}
Andr{\'a}s Gily{\'e}n, Yuan Su, Guang~Hao Low, and Nathan Wiebe, \emph{Quantum singular value transformation and beyond: exponential improvements for quantum matrix arithmetics}, arXiv:1806.01838 (2018).

\bibitem[GSLW19]{GilyenSuLowEtAl2019}
\bysame, \emph{Quantum singular value transformation and beyond: exponential improvements for quantum matrix arithmetics}, Proceedings of the 51st Annual ACM SIGACT Symposium on Theory of Computing, 2019, pp.~193--204.

\bibitem[Haa19]{Haah2019}
J.~Haah, \emph{Product decomposition of periodic functions in quantum signal processing}, Quantum \textbf{3} (2019), 190.

\bibitem[Hig02]{higham2002accuracy}
Nicholas~J Higham, \emph{Accuracy and stability of numerical algorithms}, SIAM, 2002.

\bibitem[KAHS99]{kailath1999fast}
Thomas Kailath and eds. Ali H.~Sayed, \emph{Fast reliable algorithms for matrices with structure}, SIAM, 1999.

\bibitem[LC17]{LowChuang2017}
Guang~Hao Low and Isaac~L. Chuang, \emph{Optimal hamiltonian simulation by quantum signal processing}, Phys. Rev. Lett. \textbf{118} (2017), 010501.

\bibitem[LNY23]{li2023efficient}
Haoya Li, Hongkang Ni, and Lexing Ying, \emph{On efficient quantum block encoding of pseudo-differential operators}, Quantum \textbf{7} (2023), 1031.

\bibitem[LT20a]{LinTong2020a}
Lin Lin and Yu~Tong, \emph{Near-optimal ground state preparation}, Quantum \textbf{4} (2020), 372.

\bibitem[LT20b]{LinTong2020}
\bysame, \emph{Optimal quantum eigenstate filtering with application to solving quantum linear systems}, Quantum \textbf{4} (2020), 361.

\bibitem[MRTC21]{MartynRossiTanEtAl2021}
John~M Martyn, Zane~M Rossi, Andrew~K Tan, and Isaac~L Chuang, \emph{Grand unification of quantum algorithms}, PRX Quantum \textbf{2} (2021), no.~4, 040203.

\bibitem[NSD24]{niu2024quantum}
Yuezhen Niu, Vadim Smelyanskiy, and Yulong Dong, \emph{Quantum signal processing methods and systems for composite quantum gate calibration}, May~30 2024, US Patent App. 18/319,947.

\bibitem[Ram71]{ramos1971roundoff}
George~U Ramos, \emph{Roundoff error analysis of the fast fourier transform}, mathematics of computation \textbf{25} (1971), no.~116, 757--768.

\bibitem[SK95]{sayed1995fast}
Ali~H Sayed and Thomas Kailath, \emph{Fast algorithms for generalized displacement structures and lossless systems}, Linear algebra and its applications \textbf{219} (1995), 49--78.

\bibitem[Sv97]{stewart1997stability}
M.~Stewart and P.~vanDooren, \emph{Stability issues in the factorization of structured matrices}, Siam Journal on Matrix Analysis and Applications \textbf{18} (1997), no.~1, 104--118.

\bibitem[WDL22]{WangDongLin2022}
Jiasu Wang, Yulong Dong, and Lin Lin, \emph{On the energy landscape of symmetric quantum signal processing}, Quantum \textbf{6} (2022), 850.

\bibitem[Yin22]{Ying2022}
Lexing Ying, \emph{Stable factorization for phase factors of quantum signal processing}, Quantum \textbf{6} (2022), 842.

\end{thebibliography}

\end{document}